\documentclass[12pt]{indmath}

\usepackage{neubauer}
\usepackage{graphicx}
\usepackage{amsmath}
\usepackage{physics}
\usepackage{xcolor}
\usepackage{caption}
\usepackage{soul,url}
\usepackage{subcaption}
\begin{document}

\title{A self-adjoint Fourier-type model for the iQuad wavefront sensor}
\author{Victoria Laidlaw$^1$, Oliver Fauvarque$^2$, Alfred Miksch$^3$, Benoit Neichel$^4$ and Ronny Ramlau$^{1,3}$\\\footnotesize{\rm $^1$Industrial
 Mathematics Institute, Johannes Kepler University Linz, A-4040 Linz,
 Austria.
}
\\\footnotesize{\rm $^2$IFREMER, Laboratoire Detection, Capteurs et Mesures (LDCM), Centre Bretagne, Plouzane, France.}
\\\footnotesize{\rm $^3$Johann Radon Institute for Computational and Applied Mathematics, Linz, Austria.}
\\\footnotesize{\rm $^4$Laboratoire d'Astrophysique de Marseille, France.}
\\ Corresponding author: victoria.laidlaw@indmath.uni-linz.ac.at}
\maketitle


\begin{abstract}

Advanced adaptive optics (AO) systems can use Fourier-type wavefront sensing to correct optical distortions encountered in ground-based telescopes, AO-assisted retinal imaging, and free-space optical communications (FSOC). Recently, a novel Fourier-type wavefront sensor (WFS) known as the iQuad WFS has been introduced. Its design features a focal plane tessellation with a four-quadrant phase mask (FQPM) that incorporates a  $\pm \pi/2$ phase shift between adjacent quadrants.

In this work, we establish a comprehensive mathematical framework for the iQuad WFS, including its forward models and linearizations based on the Fr\'echet derivative. We reveal a connection between the iQuad WFS and the 2d finite Hilbert transform and demonstrate that the linear iQuad WFS operator is self-adjoint - a unique property among Fourier-type WFSs. 
Additionally, we introduce the double iQuad WFS, a two-path configuration that combines two rotated iQuad WFSs. This design addresses the limitations of the single iQuad WFS by suppressing poorly-seen phase components. Moreover, the double setup simplifies the mathematical modeling. We also highlight iQuad similarities to the widely used pyramid wavefront sensor (PWFS). Finally, we extend the concept of modulation to the iQuad WFS, further enhancing its versatility.

The theoretical analysis presented here lays the groundwork for the development of fast and robust model-based wavefront reconstruction algorithms for the iQuad WFS, paving the way for future applications in AO instruments.

\end{abstract}


\section{Introduction}\label{sec:intro}

Adaptive optics (AO) systems are utilized in various instruments in ophthalmic imaging, microscopy, free-space optical communications (FSOC) and in ground-based astronomical observatories, including the next generation of extremely large telescopes (ELTs)~\cite{Roddier, Tyson91}. In AO instruments, deformable mirrors (DMs) are used to improve the image quality by mechanically correcting dynamic wavefront aberrations in real-time. The basic idea is to reflect the distorted wavefronts on a mirror that is shaped appropriately so that the corrected wavefronts allow for high image quality when observed by the science camera. Optimal positioning of mirror actuators implies the knowledge of the incoming wavefronts. Although wavefronts cannot be measured directly, they can be reconstructed with the use of a wavefront sensor (WFS), which measures the time-averaged characteristics of the captured light that is related to the incoming phase. Reconstruction of the incoming wavefront from given WFS measurements is an \textit{inverse problem}.

Different WFSs can be employed and one of the most popular choices is the Shack-Hartmann (SH) sensor~\cite{Shack1971}. However, Fourier-type WFSs such as the pyramid wavefront sensor (PWFS) have many advantages compared to the SH WFS, e.g., regarding noise propagation or sampling flexibility \cite{Fauv16}, and thus have become the sensors of choice for many AO instruments. They utilize optical Fourier filtering, achieved through a specific optical element (e.g., a multi-facet glass prism) positioned in the focal plane, see Fig.~\ref{fig_fwfs}, left. After being Fourier filtered, the light is refocused into a subsequent pupil plane, where the modified wavefront produces intensity patterns which are captured by a camera. By analyzing the patterns, the wavefront phase can be reconstructed using mathematical algorithms from the field of inverse problems.

The focus of this article is on a new Fourier-type WFS, called the iQuad WFS that is derived from the four-quadrant phase mask (FQPM) coronagraph~\cite{Rouan}.

Section~\ref{4qpm} introduces the concept of WFSs derived from coronagraphy, with a focus on the FQPM coronagraph. In Section~\ref{mathmodels}, the mathematical modeling and analysis of the iQuad WFS are presented. Section~\ref{double_iQuad} explores the double iQuad WFS, a two-path variation of the sensor. Connections and similarities with the well-known PWFS are examined in Section~\ref{sim_pwfs}. Finally, Section~\ref{mod_iquad} introduces the concept of modulation for iQuad wavefront sensing.


\section{Four-quadrant Fourier-type WFS class}\label{4qpm}

WFSs are used to obtain phase information of an incoming wavefront from intensity measurements. To this end, Fourier-type WFSs as well as coronagraphs utilize optical Fourier filtering. This technique manipulates light in the spatial frequency domain by using a mask placed in the focal plane. In coronagraphy, these masks are designed to reject light, whereas in wavefront sensing they are engineered to convert phase fluctuations into intensity fluctuations.

Certain coronagraph masks can be adapted for wavefront sensing with minor modifications. For example, the Roddier \& Roddier coronagraph~\cite{RetR} and the Zernike WFS~\cite{Ndia13,Zer34} both use a mask featuring a small circular phase-shifting spot in its center. For an optical system operating at a wavelength $\lambda$, the depth of this spot is $\lambda/2$ in the coronagraph configuration, while it is $\lambda/4$ for the Zernike WFS. Physically, a depth of $\lambda/2$ introduces a phase shift of $\pi$ between the fields inside and outside the central spot, placing them in opposition of phase. This phase opposition results in the destructive interference required for coronagraphy. In contrast, for the Zernike WFS, the phase shift is $\pi/2$, which means that the fields are in phase quadrature. Another example of a coronagraph idea being turned into a WFS is the newly introduced bivortex WFS~\cite{Chambouleyron2024CoronagraphWFS}.

In this work, we propose the extension of this concept to another coronagraph, the FQPM coronagraph~\cite{Rouan}. In the wavefront sensing context, the general idea follows those of Fourier-type WFSs with a mask having a Cartesian structure in the focal plane~\cite{FaHuShaRaLAM19_AO4ELTproc}. These masks are divided into $4$ quadrants around the origin and have an arbitrary shift $\delta \in [-\lambda/2,\lambda/2]$ in between. The corresponding optical transfer function (OTF) is given by
\begin{equation}\label{otf_fqpm}
OTF_{\text{FQPM}}(x,y) := \begin{cases}
            \exp{\frac{2i\pi}{\lambda}\delta}\, , &xy <0\, ,\\
            1\, , &\text{else}\, ,
        \end{cases}
\end{equation}
and visualized in Fig.~\ref{fig_fwfs}, right.
\begin{figure}
  \centering
  \includegraphics[width=0.9\textwidth]{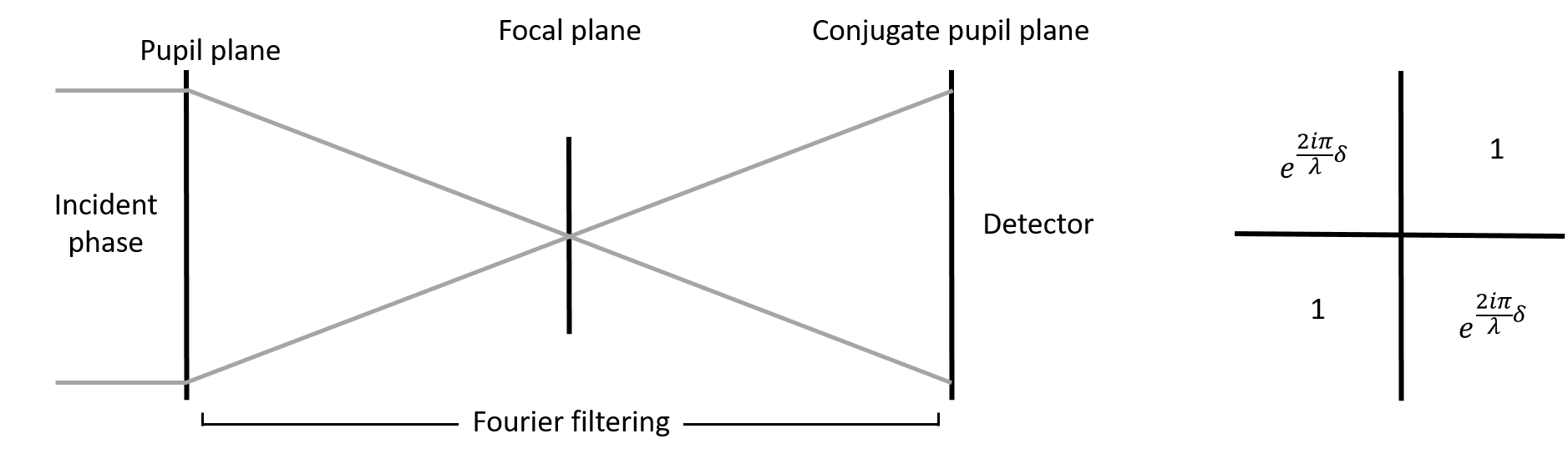}
  \caption{Concept of Fourier-type WFSs~\cite{HuNeuSha_2023} (left) and transparency function of the FQPM with an arbitrary differential piston $\delta$ (right).}
  \label{fig_fwfs}
\end{figure}

For the underlying mathematical model, a phase entering the system is represented by a real-valued 2d function $\phi(x,y)$. For simplification, a constant amplitude across the full pupil $\Omega \subset \R^2$ is assumed. Thus, for wavefront sensing we consider the incident field $\chi_\Omega e^{-i\phi}$ with the characteristic function $\chi_\Omega$ of the pupil $\Omega$. Furthermore, it is assumed that the incoming light is monochromatic and 
\begin{equation*}
 \phi(x,y) = \frac{2\pi}{\lambda}\Delta(x,y)   
\end{equation*}
is the perturbed phase at the considered wavelength $\lambda$. The optical path difference $\Delta$ was created by atmospheric turbulence or another source of perturbation. The pupil-plane intensities of the WFS are recorded on a detector $D \subset \R^2$ with $D$ bounded and $D\supseteq\Omega$.

\begin{theorem}
The propagation of a phase $\phi$ through an FQPM WFS with an arbitrary differential piston $\delta \in [-\lambda/2,\lambda/2]$ can be described by the intensity 
\begin{equation*}
    I(\phi)(x,y)=\chi_D\left(x,y\right)\left|A(\phi)(x,y)\right|^2\, ,
\end{equation*}
where the corresponding electric field $A$ is given by
\begin{equation}\label{fqpm_general}
    A(\phi)(x,y) = e^{\frac{i\pi}{\lambda}\delta}\left[\cos\left(\frac{\pi\delta}{\lambda}\right)\chi_\Omega e^{-i\phi}+i\sin\left(\frac{\pi\delta}{\lambda}\right)H(\chi_\Omega e^{-i\phi})\right](x,y)
\end{equation}
and $H$ is the 2d Hilbert transform
\begin{equation*}
    H(\phi)(x,y):=\frac{1}{\pi^2}\, p.v. \int_{\R^2} \frac{\phi(x',y')}{(x'-x)(y'-y)} \, dx'dy'
\end{equation*}
in the sense of the principal value $p.v.$. 
\end{theorem}

\begin{proof}
According to~\cite{Fauv16}, the intensity on a detector $D\subset\R^2$ of an arbitrary Fourier-type WFS is described by
\begin{equation}\label{eq_fouriertype}
     I(\phi)(x,y)=\chi_D\left(x,y\right) \abs{\mathcal{F}^{-1}\left(OTF\cdot \mathcal{F}\left(\chi_\Omega e^{-i\phi}\right)\right)(x,y)}^2\, .
\end{equation}
For an FQPM WFS, the OTF is given in eq.~\eqref{otf_fqpm}.

Using $H(\phi)(x,y) = - \mathcal{F}^{-1}\left[\text{sgn}_x \text{sgn}_y\mathcal{F}\left(\phi\right)\right](x,y)$ from~\cite{SteinWeiss1971} and the identities
    \begin{equation*}
        1 + e^{i\alpha} = 2e^{i\alpha/2}\cos{\left(\frac{\alpha}{2}\right)}\, , \quad 1 - e^{i\alpha} = -2ie^{i\alpha/2}\sin{\left(\frac{\alpha}{2}\right)}\, ,
    \end{equation*}
    yields
    \begin{align*}
        \mathcal{F}^{-1}\left[OTF_{\text{FQPM}} \cdot \mathcal{F}(\chi_\Omega e^{-i\phi})\right] &\overset{\eqref{otf_fqpm}}{=} \mathcal{F}^{-1}\left[\left(\chi _{\{xy>0\}}  + \chi_{\{xy<0\}}  e^{\frac{2i\pi}{\lambda}\delta}\right)\mathcal{F}\left({\chi_\Omega e^{-i\phi}}\right)\right] \\ &= \frac{1}{2}\mathcal{F}^{-1}\left[\left(\left(1+e^{\frac{2i\pi}{\lambda}\delta}\right) + \left(1-e^{\frac{2i\pi}{\lambda}\delta}\right) \text{sgn}_x \text{sgn}_y \right)\mathcal{F}\left({\chi_\Omega e^{-i\phi}}\right)\right]\\
        &= \frac{1}{2}\left[\left(1+e^{\frac{2i\pi}{\lambda}\delta}\right)\chi_\Omega e^{-i\phi} - \left(1-e^{\frac{2i\pi}{\lambda}\delta}\right) H\left(\chi_\Omega e^{-i\phi}\right)\right]\\
        &= e^{\frac{i\pi}{\lambda}\delta}\left[\cos\left(\frac{\pi\delta}{\lambda}\right)\chi_\Omega e^{-i\phi}+i\sin\left(\frac{\pi\delta}{\lambda}\right)H(\chi_\Omega e^{-i\phi})\right]\, .
    \end{align*}
\hfill \end{proof}

The linear field operator $A$ in eq.~\eqref{fqpm_general} consists of two terms. The first one reproduces the incoming field and the second provides its 2d Hilbert transform (up to a factor). The Hilbert transform operator $H$ is involutive and conserves energy, i.e., its $L^2$-norm. The differential piston $\delta$ acts as a regulator between the two contributions. For $\delta =0$ the corresponding FQPM is pointless since the propagator operator only contains the identity operator. The coronagraph case $\delta = \lambda/2$ gives a pure 2D Hilbert transform. If $\delta = \pm \lambda/4$, there is a uniform energy distribution between the two terms. This configuration corresponds to the iQuad WFS~\cite{FaHuShaRaLAM19_AO4ELTproc}. The notation \textit{Quad} refers to the Cartesian tessellation of the focal plane, whereas $i$ emphasizes the transformation of the coronagraph FQPM transparency function into the iQuad WFS mask (coefficients of the OTF equal to 1 and $i$ for $\delta =\lambda/4$ in Fig.~\ref{fig_fwfs}, right).

\section{Mathematical modeling of the iQuad WFS}\label{mathmodels}

We now consider the positive differential piston $\delta = \lambda/4$ and derive from eq.~\eqref{fqpm_general} the optical propagator
\begin{equation}\label{Aplus}
    A^+=e^{\frac{i\pi}{4}}\left[\frac{\mathcal{I}+i H}{\sqrt{2}}\right]\, ,
\end{equation}
where $\mathcal{I}$ represents the identity operator. The case $\delta=-\frac{\lambda}{4}$ corresponding to its conjugate, the $-$iQuad WFS, will be investigated in Section~\ref{double_iQuad}.

\begin{lemma}
    For the intensity $I$ of the iQuad WFS which is given by
    \begin{equation*}
        I(\phi)(x,y) = \chi_D\left(x,y\right)\left|A^+\left(\chi_\Omega e^{-i\phi}\right)(x,y)\right|^2
    \end{equation*}
    holds
    \begin{equation}\label{i1}
    I(\phi)(x,y) 
    = \chi_D\left(x,y\right)\left[Im\left(\left(\chi_\Omega e^{-i\phi}\right)H\left(\chi_\Omega e^{i\phi}\right)\right)+\frac{\left|H\left(\chi_\Omega e^{i\phi}\right)\right|^2+\chi_\Omega^2}{2}\right](x,y)\, .
\end{equation}
\end{lemma}
\begin{proof}
\begin{equation*}
    I(\phi)(x,y) = \chi_D\left(x,y\right)\left|A^+\left(\chi_\Omega e^{-i\phi}\right)(x,y)\right|^2
\end{equation*}
with
    \begin{align*}
    \left|A^+\left(\chi_\Omega e^{-i\phi}\right)(x,y)\right|^2
    &=\left|e^{\frac{i\pi}{4}}\right|^2 \cdot\left| \left[\frac{\chi_\Omega e^{-i\phi}+iH(\chi_\Omega e^{-i\phi})}{\sqrt{2}}\right](x,y)\right|^2 \\
    &= \frac{1}{2}\cdot\left| \left[ \chi_\Omega\left(\cos\left[\phi\right]-i\sin\left[\phi\right]\right)+iH\left(\chi_\Omega\left(\cos\left[\phi\right]-i\sin\left[\phi\right]\right)\right)  \right](x,y) \right|^2 \\
     &= \frac{1}{2}\cdot\left| \left[ \left(\chi_\Omega\cos\left[\phi\right]+H\left(\chi_\Omega\sin\left[\phi\right]\right) \right)+i\left(H\left(\chi_\Omega\cos\left[\phi\right]\right)-\chi_\Omega\sin\left[\phi\right] \right) \right](x,y) \right|^2 \\
     &= \frac{1}{2}\cdot \left[ \left(\chi_\Omega\cos\left[\phi\right]+H\left(\chi_\Omega\sin\left[\phi\right]\right) \right)^2+\left(H\left(\chi_\Omega\cos\left[\phi\right]\right)-\chi_\Omega\sin\left[\phi\right] \right)^2 \right](x,y)  \\
     &= \frac{1}{2}\cdot \left[ \chi_\Omega^2\cos^2\left[\phi\right]+2\chi_\Omega\cos\left[\phi\right]H\left(\chi_\Omega\sin\left[\phi\right]\right)+H^2\left(\chi_\Omega\sin\left[\phi\right]\right) \right.\\
     &\left.+H^2\left(\chi_\Omega\cos\left[\phi\right]\right)-2\chi_\Omega\sin\left[\phi\right]H\left(\chi_\Omega\cos\left[\phi\right]\right)+\chi_\Omega^2\sin^2\left[\phi\right] \right](x,y)  \\
      &= \frac{1}{2}\cdot \left[\underbrace{ \chi_\Omega^2\cos^2\left[\phi\right]+\chi_\Omega^2\sin^2\left[\phi\right]}_{=\chi_\Omega^2}+\underbrace{H^2\left(\chi_\Omega\cos\left[\phi\right]\right)+H^2\left(\chi_\Omega\sin\left[\phi\right]\right)}_{=\left|H(\chi_\Omega e^{i\phi})\right|^2} \right.\\
     &\left.+\underbrace{2\chi_\Omega\cos\left[\phi\right]H\left(\chi_\Omega\sin\left[\phi\right]\right)-2\chi_\Omega\sin\left[\phi\right]H\left(\chi_\Omega\cos\left[\phi\right]\right)}_{=2 Im\left(\left(\chi_\Omega e^{-i\phi}\right) H(\chi_\Omega e^{i\phi})\right)}  \right](x,y)\\
     &= \left[Im\left(\left(\chi_\Omega e^{-i\phi}\right)H\left(\chi_\Omega e^{i\phi}\right)\right)+\frac{\left|H\left(\chi_\Omega e^{i\phi}\right)\right|^2+\chi_\Omega^2}{2}\right](x,y)\, .
\end{align*}
\hfill\end{proof}

The above proof demonstrates that, for the iQuad WFS, the majority of phase information is concentrated within the pupil support $\Omega$, with a smaller portion leaking into the surrounding region bounded by the detector $D$ (see Fig.~\ref{ref_intensities}). For simplicity of notation, we omit the function support, but all functions are compactly supported by the finite pupil and detector. Consequently, while we write the norms in $L^2(\R^2)$, they are equivalent to $L^2(\Omega)$ (or $L^2(D)$) norms given finite support.

In~\cite{Fauv16} the authors introduced a numerical quantity, called meta-intensity, that describes the response of a Fourier-type WFS to an incoming turbulent phase $\phi$ around the sensor's operating point $\phi_r$. A closed AO loop will try to maintain the detector intensity as close as possible to the reference intensity $I(\phi_r)$. Conventionally, the operating point is set to $\phi_r=0$. Fig.~\ref{ref_intensities} shows an iQuad WFS standard intensity, reference intensity and the corresponding meta-intensity for an example wavefront. 

\begin{figure}[ht]
    \centering
    \includegraphics[width=0.24\linewidth]{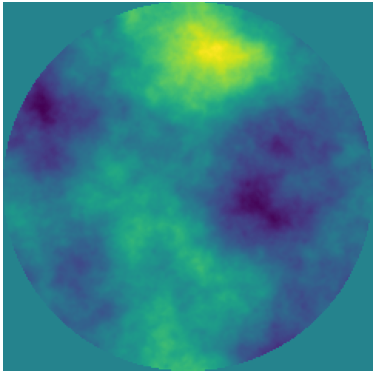}
    \includegraphics[width=0.24\linewidth]{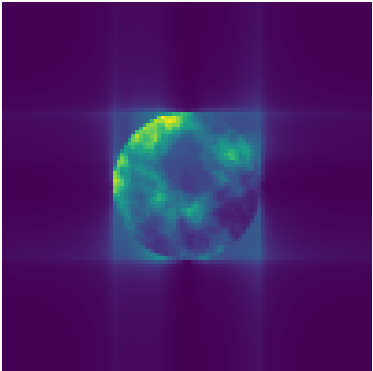}
    \includegraphics[width=0.24\linewidth]{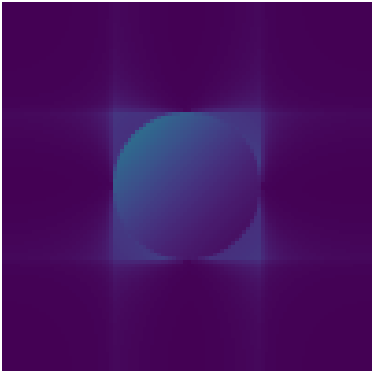}
    \includegraphics[width=0.24\linewidth]{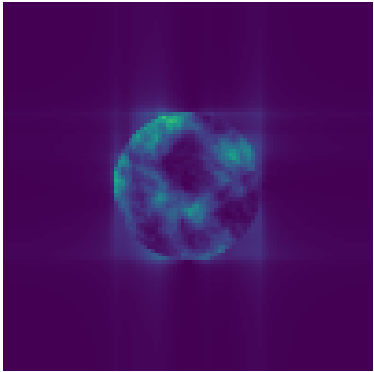}
    \caption{Example wavefront $\phi$, corresponding iQuad WFS intensity $I(\phi)$, reference intensity $I(0)$ and iQuad WFS meta-intensity $m(\phi)$ (from left to right).}
    \label{ref_intensities}
\end{figure}

\begin{definition}\label{def_i}
We define the iQuad WFS operator $\boldsymbol i$ by
$$\boldsymbol i = \boldsymbol i_1+\boldsymbol i_2$$
with
\begin{align}\label{def_i1}
\left(\boldsymbol i_1\phi\right)(x,y)&:=\chi_{\Omega}(x,y)\left.\dfrac{1}{\pi^2}\, p.v. \int_\Omega{\dfrac{\sin{\left[\phi(x',y')-\phi(x,y)\right] }}{(x'-x)(y'-y)}\, dx'\ dy'}\right.
\end{align}
and
\begin{align}\label{def_i2}
\left(\boldsymbol i_2\phi\right)(x,y)&:= \chi_{D}(x,y)\dfrac{1}{2\pi^4}\, p.v.  \int_\Omega{\int_\Omega{\dfrac{\cos{\left[\phi(x',y')-\phi(x'',y'')\right]-1 }}{(x'-x)(y'-y)(x''-x)(y''-y)} \, dx'' \ dy'' \ dy' \  } dx' }\, .
\end{align}
\end{definition}

\begin{theorem}
    The meta-intensity $m$ of the iQuad WFS for a phase $\phi$ is given by
    \begin{equation*}
       \left(m\phi\right)(x,y) = \left(\boldsymbol i\phi\right)(x,y)
    \end{equation*}
    with $i$ defined as in Definition~\ref{def_i}.
\end{theorem}

\begin{proof}
The meta-intensity for Fourier-type WFSs, like the iQuad WFS, around $\phi_r=0$ with normalized flux is given by~\cite{Fauv16}
\begin{equation*}
    \left(m\phi\right)(x,y) = I(\phi)(x,y)-I(0)(x,y)\, .
\end{equation*}
The claim of the theorem follows by developing formula~\eqref{i1} with respect to spatial coordinates. We start with the reference intensity 
\begin{align*}
    I(0)(x,y) &= \chi_{D}(x,y)\left[\underbrace{Im\left(\chi_\Omega H\left(\chi_\Omega \right)\right)}_{=0}+\frac{\left|H\left(\chi_\Omega \right)\right|^2+\chi_\Omega^2}{2}\right](x,y) \\
    &=\chi_{D}(x,y)\left[\frac{\left|H\left(\chi_\Omega \right)\right|^2+\chi_\Omega^2}{2}\right](x,y)\, .
\end{align*}
Then,
\begin{align}\label{def_i11}
    \left(m\phi\right)(x,y) &= I(\phi)(x,y)-I(0)(x,y) \nonumber \\
    &= \chi_{D}(x,y)\left[Im\left(\left(\chi_\Omega e^{-i\phi}\right)H\left(\chi_\Omega e^{i\phi}\right)\right)+\frac{\left|H\left(\chi_\Omega e^{i\phi}\right)\right|^2+\chi_\Omega^2-\left|H\left(\chi_\Omega \right)\right|^2-\chi_\Omega^2}{2}\right](x,y) \nonumber\\
    &= \chi_{D}(x,y)\left[\underbrace{Im\left(\left(\chi_\Omega e^{-i\phi}\right)H\left(\chi_\Omega e^{i\phi}\right)\right)}_{=:m_1}+\underbrace{\frac{\left|H\left(\chi_\Omega e^{i\phi}\right)\right|^2-\left|H\left(\chi_\Omega \right)\right|^2}{2}}_{=:m_2}\right](x,y)\, . 
\end{align}
As the Hilbert transform of a real function is a real function, the term $m_1$ simplifies to
\begin{align*}    
    m_1(x,y)&= \Big[Im\left(\left(\chi_\Omega\cos[\phi]-i\chi_\Omega \sin[\phi]\right)H\left(\chi_\Omega\cos[\phi]+i\chi_\Omega\sin[\phi]\right)\right)\Big](x,y) \\
    &=\Big[\chi_\Omega\cos[\phi]H\left(\chi_\Omega\sin[\phi]\right)-\chi_\Omega\sin[\phi]H\left(\chi_\Omega\cos[\phi]\right)\Big](x,y)\\
    &=\chi_\Omega(x,y)\frac{1}{\pi^2}\, p.v. \int_\Omega\frac{\cos[\phi(x,y)]\sin[\phi(x',y')]-\sin[\phi(x,y)]\cos[\phi(x',y')]}{(x'-x)(y'-y)}\, dx'\ dy' \\
    &=\chi_\Omega(x,y)\frac{1}{\pi^2}\, p.v. \int_\Omega\frac{-\sin\left[\phi(x,y)-\phi(x',y')\right]}{(x'-x)(y'-y)}\, dx'\ dy' \\
    &=\left(\boldsymbol i_1 \phi\right)(x,y)\, .
\end{align*}
For the term $m_2$ holds
\begin{align*}
    \chi_{D}(x,y)m_2(x,y)&=\chi_{D}(x,y)\frac{1}{2}\left[\left|H\left(\chi_\Omega\left(\cos[\phi]+i\sin[\phi]\right)\right)\right|^2-\left|H\left(\chi_\Omega\right)\right|^2\right](x,y) \\
    &= \chi_{D}(x,y)\frac{1}{2}\left[H^2\left(\chi_\Omega\cos[\phi]\right)+H^2\left(\chi_\Omega\sin[\phi]\right)-\left|H\left(\chi_\Omega\right)\right|^2\right](x,y) \\
    &=\chi_{D}(x,y)\frac{1}{2\pi^4}\, p.v.\int_\Omega\,\int_\Omega \frac{\cos[\phi(x',y')]\cos[\phi(x'',y'')]}{(x'-x)(y'-y)(x''-x)(y''-y)}\\
    &+ \frac{\sin[\phi(x',y')]\sin[\phi(x'',y'')]-1}{(x'-x)(y'-y)(x''-x)(y''-y)} \, dx'\ dy'\ dx''\ dy''\\
    &=\chi_{D}(x,y)\frac{1}{2\pi^4}\, p.v.\int_\Omega\,\int_\Omega\frac{\cos[\phi(x',y')-\phi(x''-y'')]-1}{(x'-x)(y'-y)(x''-x)(y''-y)} \, dx'\ dy'\ dx''\ dy''\\
    &=\left(\boldsymbol i_2 \phi\right)(x,y)\, .
\end{align*}
Hence,
\begin{equation*}
    \left(m\phi\right)(x,y) = \chi_{D}(x,y)\left(m_1(x,y)+m_2(x,y)\right)= \left(\boldsymbol i_1\phi\right)(x,y)+\left(\boldsymbol i_2\phi\right)(x,y) = \left(\boldsymbol i\phi\right)(x,y)\, ,
\end{equation*}
which proves the claim of the theorem.
\end{proof}

For what follows, we assume the wavefront phase $\phi \in H^s(\R^2)$ with $s>1$. In astronomical AO this is typically satisfied under Kolmogorov turbulence~\cite{Roddier,RoWe96}, which implies $H^{11/6}\left(\R^2\right)$~\cite{Ellerbroek02,Eslitz13}. For ophthalmic or microscopic  AO, to the best of our knowledge, there are no established smoothness model. Hence, we assume $\phi \in H^s(\Omega)$ for some $1<s<2$.

\begin{proposition}\label{op_well_def}
The non-linear operator $\boldsymbol i : H^{s}(\R^2) \rightarrow L^2 (\R^2)$ for $1<s<2$, representing the iQuad WFS, is a well-defined operator between the given spaces above and Fr\'{e}chet differentiable.
\end{proposition}

\begin{proof}
The proof is given in~[\cite{HuNeuSha_2023}, Remark 2.2] for $\phi \in H^s$ with $s>1$ and $OTF\in\mathcal{L}^\infty(\R^2)$.
\end{proof}

The Fr\'{e}chet derivative of the iQuad WFS operator $\boldsymbol i$ is essential for the calculation of linear approximations or the application of nonlinear iterative algorithms from the field of inverse problems.

\begin{theorem}\label{3.9_gateaux}
The Fr\'{e}chet derivative $\boldsymbol i'(\phi) \in  \mathcal{L}\left(H^{s},L^2\right)$, $1<s<2$, at $\phi \in \mathcal{D}\left(\boldsymbol i\right)$ of the non-linear iQuad WFS operator $\boldsymbol i: \mathcal{D}\left(\boldsymbol i\right)\subseteq H^{s}\left(\R^2\right)\rightarrow L^2\left(\R^2\right)$ is given by
\begin{equation*}\label{theo_gateaux_def1}
\left(\boldsymbol i'(\phi)\ \psi\right)\left(x,y\right) = \left(\boldsymbol i_1'(\phi)\ \psi\right)\left(x,y\right) +\left(\boldsymbol i_2'(\phi)\ \psi\right)\left(x,y\right) 
\end{equation*}
with
\begin{equation*}
\left(\boldsymbol i_1'(\phi)\ \psi\right)\left(x,y\right) 
=\chi_\Omega\left(x,y\right)\dfrac{1}{\pi^2}  \, p.v.\int_\Omega{\dfrac{\cos{\left[\phi(x',y')-\phi(x,y)\right]}\left[\psi(x',y')-\psi(x,y)\right]}{(x'-x)(y'-y)}\, dx' \ dy'}
\end{equation*}
and
\begin{align*}
\left(\boldsymbol i_2'(\phi)\ \psi\right)\left(x,y\right) 
&=-\chi_D\left(x,y\right)\dfrac{1}{2\pi^4} \, p.v. \int_\Omega\int_\Omega \sin{\left[\phi(x',y')-\phi(x'',y'')\right]}\\
&\cdot {\dfrac{\left[\psi(x',y')-\psi(x'',y'')\right]}{(x'-x)(y'-y)(x''-x)(y''-y)}\, dx'' \ dy''\ dy' \ dx'}\, .
\end{align*}
\end{theorem}

\begin{proof}
We know from Proposition~\ref{op_well_def} that the iQuad WFS operator $\boldsymbol i$ is Fr\'{e}chet differentiable. Hence, the Fr\'{e}chet derivative must coincide with the G\^{a}teaux derivative and it suffices to calculate the latter.

We utilize the representation (cf Taylor's theorem with Lagrange form of the remainder)
\begin{equation}\label{eq:taylorsin}
\begin{split}
\sin\left(\phi+\psi\right) &= \sin\left(\phi\right)+\sin'\left(\phi\right)\psi+\dfrac{1}{2}\sin''\left(\phi+\theta\psi\right) \psi^2 \\
&= \sin\left(\phi\right)+\cos\left(\phi\right)\psi-\dfrac{1}{2}\sin\left(\phi+\theta\psi\right) \psi^2\, , \\
\cos\left(\phi+\psi\right) &= \cos\left(\phi\right)+\cos'\left(\phi\right)\psi+\dfrac{1}{2}\cos''\left(\phi+\theta\psi\right) \psi^2 \\
&=\cos\left(\phi\right)-\sin\left(\phi\right)\psi-\dfrac{1}{2}\cos\left(\phi+\theta\psi\right) \psi^2
\end{split}
\end{equation}
for a $\theta = \theta\left(\phi,\psi\right) \in \left(0,1\right)$.

The derivatives $\boldsymbol i_1'(\phi)$ and $\boldsymbol i_2'(\phi)$ compute as
{\small \begin{equation*} \label{gateaux_deriv1}
\begin{split}
\left(\boldsymbol i_1'(\phi)\ \psi\right)\left(x,y\right)
&=\lim\limits_{t \rightarrow 0} \dfrac{\left(\boldsymbol i_1\left(\phi+t\psi\right)\right)\left(x,y\right)-\left(\boldsymbol i_1\phi\right)\left(x,y\right)}{t} \\
&=\lim\limits_{t\rightarrow 0} \chi_\Omega\left(x,y\right)\dfrac{1}{\pi^2} \, p.v.\int_\Omega
\dfrac{\sin {\left[\phi(x',y')+t\psi(x',y')-\phi(x,y)-t\psi(x,y)\right]}}{t(x'-x)(y'-y)} \\
&-\dfrac{\sin{\left[\phi(x',y')-\phi(x,y)\right]}}{t(x'-x)(y'-y)}\, dx' \ dy' \\ 
&\overset{\eqref{eq:taylorsin}}{=}\lim\limits_{t\rightarrow 0}\chi_\Omega\left(x,y\right)  \dfrac{1}{\pi^2} \, p.v. \int_\Omega\dfrac{\sin'{\left[\phi(x',y')-\phi(x,y)\right]}\left[t\psi(x',y')-t\psi(x,y)\right]}{t(x'-x)(y'-y)} \\ 
&+\dfrac{\dfrac{1}{2}\sin''{\left[\phi(x',y')-\phi(x,y)+\theta t \left[\psi(x',y')-\psi(x,y)\right]\right]}t^2\left[\psi(x',y')-\psi(x,y)\right]^2}{t(x'-x)(y'-y)}\, dx' \ dy' \\
&=\chi_\Omega\left(x,y\right)\dfrac{1}{\pi^2} \, p.v. \int_\Omega{\dfrac{\cos{\left[\phi(x',y')-\phi(x,y)\right]}\left[\psi(x',y')-\psi(x,y)\right]}{(x'-x)(y'-y)}\, dx' \ dy'}\, ,
\end{split}
\end{equation*}
\begin{equation*} \label{gateaux_deriv2}
\begin{split}
\left(\boldsymbol i_2'(\phi)\ \psi\right)\left(x,y\right)
&=\lim\limits_{t \rightarrow 0} \dfrac{\left(\boldsymbol i_2\left(\phi+t\psi\right)\right)\left(x,y\right)-\left(\boldsymbol i_2\phi\right)\left(x,y\right)}{t} \\
&=\lim\limits_{t\rightarrow 0}\chi_D\left(x,y\right) \dfrac{1}{2\pi^4} \, p.v.\left(\int_\Omega\int_\Omega\dfrac{\cos {\left[\phi(x',y')+t\psi(x',y')-\phi(x'',y'')-t\psi(x'',y'')\right]}}{t(x'-x)(y'-y)(x''-x)(y''-y)}\right.\\
&\left.-\dfrac{1}{t(x'-x)(y'-y)(x''-x)(y''-y)}\, dx'' \ dy'' \ dy' \ dx'\right. \\ 
&\left.-\int_\Omega\int_\Omega{\dfrac{\cos{\left[\phi(x',y')-\phi(x'',y'')\right]}-1}{t(x'-x)(y'-y)(x''-x)(y''-y)}\, dx'' \ dy'' \ dy' \ dx'}\right) \\ 
&\overset{\eqref{eq:taylorsin}}{=}\lim\limits_{t\rightarrow 0}\chi_D\left(x,y\right)  \dfrac{1}{2\pi^4}\, p.v.\bigg(  \int_\Omega\int_\Omega\dfrac{\cos'{\left[\phi(x',y')-\phi(x'',y'')\right]}\left[t\psi(x',y')-t\psi(x'',y'')\right]}{t(x'-x)(y'-y)(x''-x)(y''-y)}\bigg. \\ 
&\bigg.+\dfrac{\dfrac{1}{2}\cos''{\left[\phi(x',y')-\phi(x'',y'')+\theta t \left[\psi(x',y')-\psi(x'',y'')\right]\right]}}{t(x'-x)(y'-y)(x''-x)(y''-y)}\bigg. \\
&\bigg.\cdot t^2\left[\psi(x',y')-\psi(x'',y'')\right]^2 \, dx'' \ dy'' \ dy' \ dx'\bigg) \\
&=-\chi_D\left(x,y\right)\dfrac{1}{2\pi^4} \\
&\cdot \, p.v.\int_\Omega\int_\Omega{\dfrac{\sin{\left[\phi(x',y')-\phi(x'',y'')\right]}\left[\psi(x',y')-\psi(x'',y'')\right]}{(x'-x)(y'-y)(x''-x)(y''-y)}\, dx'' \ dy''\ dy' \ dx'} \,.
\end{split}
\end{equation*}
}
\hfill\end{proof}

In closed-loop AO, the iQuad WFS measures already corrected residual wavefronts. Hence, in a stable system, the residual phases $\phi$ are relatively small and a linear approximation of the WFS operator around $\phi = 0$ is typically sufficient.

\begin{theorem}\label{3.9}
The linearization $\boldsymbol i^{lin}:H^{s}\left(\R^2\right)\rightarrow L^2\left(\R^2\right)$, $1<s<2$ by means of the Fr\'{e}chet derivative of the operator $\boldsymbol i$ introduced in Definition~\ref{def_i} around $0$ is given by
\begin{equation}\label{liniQuad}
\left(\boldsymbol i^{lin}\phi\right)(x,y)=\chi_{\Omega}(x,y)\dfrac{1}{\pi^2}\, p.v.  \int_\Omega{\dfrac{\phi(x',y')-\phi(x,y)}{(x'-x)(y'-y)}\, dx' \ dy'}\, .
\end{equation}
\end{theorem}

\begin{proof}
The claim immediately follows from $\left(\boldsymbol i^{lin}\phi\right):= \left(  \boldsymbol i'  (0) \ \phi\right)$, i.e., considering the Fr\'{e}chet derivative from Theorem~\ref{3.9_gateaux} in $\phi=0$ and in direction $\psi=\phi$.
\end{proof}

The linear term $i^{lin}$ exhibits some interesting properties about its support. Specifically, eq.~\eqref{liniQuad} reveals that its support is precisely aligned with the pupil $\Omega$. In other words, there is no linear information outside the pupil region. When linear reconstruction algorithms are applied for wavefront reconstruction, it can thus be beneficial to cut off the intensity information outside the pupil $\Omega$.

\bigskip\par

In order to reconstruct the wavefront from iQuad data, knowledge of the adjoint of the Fr\'{e}chet derivative can be required.
As discussed in Proposition~\ref{op_well_def}, it is assumed that the phase $\phi$ belongs to the Sobolev space $H^s(\R^2)$ with a smoothness index $1<s<2$. However, the WFS data only belong to $L^2$. In the following, we use the embedding operator $E_s$ as investigated in Ref.~\cite{Teschke2004}, defined by
\begin{align*}
    E_s:H^s\left(\R^2\right)&\rightarrow L^2\left(\R^2\right)\, , \\
    E_s\phi &= \phi\, .
\end{align*}
The operator $E_s$ is continuous. If $\phi \in H^s(\Omega)$ with $\Omega$ bounded, then $E_s$ is compact.
For $\tilde{\boldsymbol{i}}:L^2(\R^2)\rightarrow L^2(\R^2)$ $$ \left(\boldsymbol i\phi\right) =: \left(\tilde{\boldsymbol{i}}\left(E_s\phi\right)\right) \qquad \implies \qquad \boldsymbol i^* = E_s^*\circ \tilde{\boldsymbol{i}}^*\, .$$
The described Sobolev embedding operator and its adjoint are used in the following proposition.
\begin{proposition}\label{roof_frechet_adjoint}
The adjoint operator $\boldsymbol i'(\phi)^*:L^2\left(\R^2\right)\rightarrow H^{s}\left(\R^2\right)$, $1<s<2$, of the iQuad sensor Fr\'{e}chet derivative in $\phi$ is represented by
\begin{equation*}\label{frech_def_a1}
\left(\boldsymbol i'(\phi)^*\varphi\right)\left(x,y\right)
=E_s^*\left(\left(\tilde{\boldsymbol i_1}'(\phi)^* \varphi\right) +\left(\tilde{\boldsymbol i_2}'(\phi)^*\varphi\right)\right)\left(x,y\right)
\end{equation*}
with $\tilde{\boldsymbol i_1}'(\phi)^*,\tilde{\boldsymbol i_2}'(\phi)^*:L^2\left(\R^2\right)\rightarrow L^2\left(\R^2\right)$ given by
\begin{equation*}\label{frech_def_a2}
\begin{aligned}
\left(\tilde{\boldsymbol i_1}'(\phi)^*\varphi\right)\left(x,y\right) &=\chi_{\Omega}\left(x,y\right)\dfrac{1}{\pi^2} \, p.v.\int_\Omega{\dfrac{\cos{\left[\phi(x',y')-\phi(x,y)\right]}\left[\varphi(x',y')-\varphi(x,y)\right]}{(x'-x)(y'-y)}\, dx'  \ dy'}\, , \\
\left(\tilde{\boldsymbol i_2}'(\phi)^*\varphi\right)\left(x,y\right) 
&= -\chi_{\Omega}\left(x,y\right)\dfrac{1}{\pi^4}\, p.v. \int_D\int_\Omega{\dfrac{\sin{\left[\phi(x,y)-\phi(x'',y'')\right]}\varphi(x',y')}{(x'-x)(y'-y)(x''-x')(y''-y')} \, dx'' \ dy'' \ dy' \ dx'}
\end{aligned}
\end{equation*}
and the embedding operator $E_s:H^s\left(\R^2\right)\rightarrow L^2\left(\R^2\right)$ with adjoint
\begin{equation*}
\left(F\left(E_s^*\psi\right)\right)(\xi,\eta)=\left(1+\xi^2+\eta^2\right)^{-s}\left(F\left(\psi\right)\right)(\xi,\eta)\, .
\end{equation*}
\end{proposition}
\begin{proof}
Let us define the operator $\tilde{\boldsymbol i}'(\phi):L^2\left(\R^2\right)\rightarrow L^2\left(\R^2\right)$ by
\begin{equation*}
        \boldsymbol i'(\phi) =: \tilde{\boldsymbol i}'(E_s(\phi)) \, ,
    \end{equation*}
where $ \boldsymbol i'(\phi)$ is the Fr\'{e}chet derivative from Theorem~\ref{3.9_gateaux}.

The adjoint operator of the Fr\'{e}chet derivative $\boldsymbol i'(\phi)^*:L^2\left(\R^2\right)\rightarrow H^s\left(\R^2\right)$ can then be calculated by 
    \begin{equation*}
        \boldsymbol i'(\phi)^* = E_s^*\left(\tilde{\boldsymbol i}'(\phi)^*\right)\, ,
    \end{equation*}
where the adjoint of the embedding operator $E_s$ follows from~\cite{Teschke2004}.

It remains to evaluate the adjoint operator $\tilde{\boldsymbol i}'(\phi)^*:L^2\left(\R^2\right)\rightarrow L^2\left(\R^2\right)$. Therefore, we divide the Fr\'{e}chet derivative $\tilde{\boldsymbol i}'(\phi)$ into four parts $\boldsymbol i_{11}\left(\phi\right)$, $\boldsymbol i_{12}\left(\phi\right)$, $\boldsymbol i_{21}\left(\phi\right)$, $\boldsymbol i_{22}\left(\phi\right) \in \mathcal{L}\left(L^2,L^2\right)$ by 
\begin{align}
\tilde{\boldsymbol i}'(\phi) &= \tilde{\boldsymbol i_1}'(\phi) +\tilde{\boldsymbol i_2}'(\phi)\, , \nonumber \\
\tilde{\boldsymbol i_1}'(\phi) &=\boldsymbol i_{11}(\phi) - \boldsymbol i_{12}(\phi)\, ,\label{split2} \\
\tilde{\boldsymbol i_2}'(\phi) &=\boldsymbol i_{21}(\phi) - \boldsymbol i_{22}(\phi) \nonumber
\end{align}
with
\begin{align*}
\left(\boldsymbol i_{11}(\phi)\ \psi\right)\left(x,y\right)
&:=\chi_{\Omega}\left(x,y\right)\dfrac{1}{\pi^2} \, p.v.\int_\Omega{\dfrac{\cos{\left[\phi(x',y')-\phi(x,y)\right]}\psi(x',y')}{(x'-x)(y'-y)}\, dx' \ dy'}\, , \\
\left(\boldsymbol i_{12}(\phi)\ \psi\right)\left(x,y\right)
&:=\chi_{\Omega}\left(x,y\right)\dfrac{1}{\pi^2} \, p.v.\int_\Omega{\dfrac{\cos{\left[\phi(x',y')-\phi(x,y)\right]}\psi(x,y)}{(x'-x)(y'-y)}\, dx' \ dy'}\, , \\
\left(\boldsymbol i_{21}(\phi)\ \psi\right)\left(x,y\right)
&:=-\chi_{D}\left(x,y\right)\dfrac{1}{2\pi^4} \, p.v. \int_\Omega\int_\Omega{\dfrac{\sin{\left[\phi(x',y')-\phi(x'',y'')\right]}\psi(x',y')}{(x'-x)(y'-y)(x''-x)(y''-y)}\, dx'' \ dy''\ dy' \ dx'}\, , \\
\left(\boldsymbol i_{22}(\phi)\ \psi\right)\left(x,y\right)
&:=-\chi_{D}\left(x,y\right)\dfrac{1}{2\pi^4} \, p.v. \int_\Omega\int_\Omega{\dfrac{\sin{\left[\phi(x',y')-\phi(x'',y'')\right]}\psi(x'',y'')}{(x'-x)(y'-y)(x''-x)(y''-y)}\, dx'' \ dy''\ dy' \ dx'}\, .
\end{align*}

Please note that for the evaluation of the adjoint operators below, it is necessary to exchange integrals and the limit in the p.v. meaning. The corresponding theorem and proof can be found in the appendix (see Theorem~\ref{theorem: Fubini L^1(Omega^3)} and below).

For $\psi, \varphi \in L^2\left(\R^2\right)$ we consider
\begin{align*}
\langle \boldsymbol i_{11}(\phi)\ \psi,\varphi\rangle_{ L^2\left(\R^2\right)} &= 
 \int_{\R^2}\chi_{\Omega}\left(x,y\right)\\
 &\cdot \left[\dfrac{1}{\pi^2} \, p.v. \int_{\R^2}\chi_{\Omega}\left(x',y'\right){\dfrac{\cos{\left[\phi(x',y')-\phi(x,y)\right]}\psi(x',y')}{(x'-x)(y'-y)}\ dx' \ dy'}\right] \varphi\left(x,y\right) \ dx \ dy \\
&= 
\int_{\R^2}\chi_{\Omega}\left(x',y'\right)\psi\left(x',y'\right)\\
&\cdot \dfrac{1}{\pi^2}  \, p.v.\int_{\R^2}\chi_{\Omega}\left(x,y\right)\dfrac{\cos{\left[\phi(x',y')-\phi(x,y)\right]}\varphi(x,y)}{(x'-x)(y'-y) } \ dx \   dy \ dx' \ dy' \\
&= 
\int_{\R^2}\chi_{\Omega}\left(x,y\right) \psi\left(x,y\right)\\
&\cdot \dfrac{1}{\pi^2} \, p.v.\int_{\R^2}\chi_{\Omega}\left(x',y'\right)\dfrac{\cos{\left[\phi(x,y)-\phi(x',y')\right]}\varphi(x',y')}{(x-x')(y-y')}\, dx' \ dy' \ dx \ dy \\
&=
 \langle \psi,\boldsymbol i_{11}(\phi)^*\varphi\rangle_{ L^2\left(\R^2\right)} 
\end{align*}
with (as cosine is even) $$\left(\boldsymbol i_{11}(\phi)^*\varphi\right)(x,y) = \chi_{\Omega}\left(x,y\right)\dfrac{1}{\pi^2} \, p.v.\int_\Omega{\dfrac{\cos{\left[\phi(x',y')-\phi(x,y)\right]}\varphi(x',y')}{(x'-x)(y'-y)}\, dx' \ dy'}\, ,$$ i.e., $\boldsymbol i_{11}(\phi)$ is self-adjoint in $L^2(\R^2)$.

Moreover,
\begin{align*}
\langle \boldsymbol i_{12}(\phi)\ \psi,\varphi\rangle_{ L^2\left(\R^2\right)} 
&=
 \int_{\R^2}\chi_{\Omega}\left(x,y\right)\\
 &\cdot \left[\dfrac{1}{\pi^2}  \, p.v.\int_{\R^2}\chi_{\Omega}\left(x',y'\right)\dfrac{\cos{\left[\phi(x',y')-\phi(x,y)\right]}\psi(x,y)}{(x'-x)(y'-y)}\, dx' \ dy' \right] \varphi\left(x,y\right) \ dx \ dy \\
 &= 
\int_{\R^2}\chi_{\Omega}\left(x,y\right) \psi\left(x,y\right)\\
&\cdot \dfrac{1}{\pi^2} \, p.v.\int_{\R^2}\chi_{\Omega}\left(x',y'\right)\dfrac{\cos{\left[\phi(x',y')-\phi(x,y)\right]}\varphi(x,y)}{(x'-x)(y'-y)}\, dx' \  dy'  \ dx \ dy \\
&=  \langle \psi,\boldsymbol i_{12}(\phi)^*\varphi\rangle_{ L^2\left(\R^2\right)}\, , 
\end{align*}
which results in $$\left(\boldsymbol i_{12}(\phi)^*\varphi\right)(x,y) =\chi_{\Omega}\left(x,y\right) \dfrac{1}{\pi^2} \, p.v.\int_\Omega{\dfrac{\cos{\left[\phi(x',y')-\phi(x,y)\right]}\varphi(x,y)}{(x'-x)(y'-y)}\, dx' \ dy' }\, ,$$ i.e., $\boldsymbol i_{12}(\phi)$ is self-adjoint in $L^2(\R^2)$.

Similarly,
\begin{align*}
\langle \boldsymbol i_{21}(\phi)\ \psi,\varphi\rangle_{ L^2\left(\R^2\right)} &= 
- \int_{\R^2}\chi_D(x,y)\left[\dfrac{1}{2\pi^4} \, p.v. \int_{\R^2} \chi_\Omega(x',y')\int_{\R^2}\chi_\Omega(x'',y'')\right. \\
&\left. \cdot {\dfrac{\sin{\left[\phi(x',y')-\phi(x'',y'')\right]}\psi(x',y')}{(x'-x)(y'-y)(x''-x)(y''-y)} \, dx'' \ dy'' \ dy' \ dx'}\right]  \varphi\left(x,y\right) \, dy \ dx \\
&= 
-\int_{\R^2}\chi_\Omega(x',y')\psi\left(x',y'\right)\dfrac{1}{2\pi^4} \, p.v. \int_{\R^2} \chi_D(x,y)\int_{\R^2} \chi_\Omega(x'',y'') \\
&\cdot {\dfrac{\sin{\left[\phi(x',y')-\phi(x'',y'')\right]}\varphi(x,y)}{(x'-x)(y'-y)(x''-x)(y''-y) } \, dx'' \ dy'' \ dy \ }  dx \ dy' \ dx' \\
&= 
-\int_{\R^2}\chi_\Omega(x,y) \psi\left(x,y\right) \dfrac{1}{2\pi^4}\, p.v. \int_{\R^2}\chi_D(x',y')\int_{\R^2}\chi_\Omega(x'',y'')\\
&\cdot {\dfrac{\sin{\left[\phi(x,y)-\phi(x'',y'')\right]}\varphi(x',y')}{(x-x')(y-y')(x''-x')(y''-y')} \, dx'' \ dy'' \ dy'} \ dx' \ dy \ dx \\
&=
 \langle \psi,\boldsymbol i_{21}(\phi)^*\varphi\rangle_{ L^2\left(\R^2\right)} 
\end{align*}

with $$\left(\boldsymbol i_{21}(\phi)^*\varphi\right)(x,y) = -\chi_\Omega(x,y)\dfrac{1}{2\pi^4} \, p.v.\int_D\int_\Omega{\dfrac{\sin{\left[\phi(x,y)-\phi(x'',y'')\right]}\varphi(x',y')}{(x'-x)(y'-y)(x''-x')(y''-y')} \, dx'' \ dy'' \ dy' \ dx'}$$

and
\begin{align*}
\langle \boldsymbol i_{22}(\phi)\ \psi,\varphi\rangle_{ L^2\left(\R^2\right)} 
&=
- \int_{\R^2}\chi_D(x,y)\left[\dfrac{1}{2\pi^4}  \, p.v.\int_{\R^2}\chi_\Omega(x',y')\int_{\R^2}\chi_\Omega(x'',y'') \right. \\
&\left. \cdot {\dfrac{\sin{\left[\phi(x',y')-\phi(x'',y'')\right]}\psi(x'',y'')}{(x'-x)(y'-y)(x''-x)(y''-y)} \ dx'' \ dy'' \ dy' \ dx' }\right] \varphi\left(x,y\right) \, dx \ dy \\
 &= 
-\int_{\R^2} \chi_\Omega(x'',y'')\psi\left(x'',y''\right) \dfrac{1}{2\pi^4} \, p.v.\int_{\R^2}\chi_\Omega(x',y')\int_{\R^2}\chi_D(x,y) \\
&\cdot {\dfrac{\sin{\left[\phi(x',y')-\phi(x'',y'')\right]}\varphi(x,y)}{(x'-x)(y'-y)(x''-x)(y''-y)} \, dx \ dy \ dy' \  dx' } \ dx'' \ dy'' \\
 &= 
-\int_{\R^2}\chi_\Omega(x,y) \psi\left(x,y\right) \dfrac{1}{2\pi^4} \, p.v.\int_{\R^2}\chi_\Omega(x',y')\int_{\R^2}\chi_D(x'',y'') \\
&\cdot {\dfrac{\sin{\left[\phi(x',y')-\phi(x,y)\right]}\varphi(x'',y'')}{(x'-x'')(y'-y'')(x-x'')(y-y'')} \, dx'' \ dy'' \ dy' \  dx' } \ dx \ dy \\
&=  \langle \psi,\boldsymbol i_{22}(\phi)^*\varphi\rangle_{ L^2\left(\R^2\right)}\, , 
\end{align*}
which by Fubini-Tonelli (understood in the p.v. sense) results in 
\begin{equation*}
\begin{split}
\left(\boldsymbol i_{22}(\phi)^*\varphi\right)(x,y) &= - \chi_\Omega(x,y)\dfrac{1}{2\pi^4} \, p.v.\int_D\int_\Omega{\dfrac{\sin{\left[\phi(x',y')-\phi(x,y)\right]}\varphi(x'',y'')}{(x''-x)(y''-y)(x''-x')(y''-y')} \, dx' \ dy' \ dy'' \ dx'' } \\
&= -\chi_\Omega(x,y)\dfrac{1}{2\pi^4}\, p.v. \int_D\int_\Omega{\dfrac{\sin{\left[\phi(x'',y'')-\phi(x,y)\right]}\varphi(x',y')}{(x'-x)(y'-y)(x'-x'')(y'-y'')} \, dx'' \ dy'' \ dy' \ dx' } \\
&= \chi_\Omega(x,y)\dfrac{1}{2\pi^4} \, p.v.\int_D\int_\Omega{\dfrac{\sin{\left[\phi(x,y)-\phi(x'',y'')\right]}\varphi(x',y')}{(x'-x)(y'-y)(x''-x')(y''-y')} \, dx'' \ dy'' \ dy' \ dx' } \, . 
\end{split}
\end{equation*}
Hence, the $L^2$-adjoint of the Fr\'{e}chet derivative $\tilde{\boldsymbol i}'(\phi)^*$ is given by
\begin{align*}
\left(\tilde{\boldsymbol i}'(\phi)^*\varphi\right)\left(x,y\right) &=\left(\boldsymbol i_{11}(\phi)^*\varphi\right)\left(x,y\right)-\left(\boldsymbol i_{12}(\phi)^*\varphi\right)\left(x,y\right)+\left(\boldsymbol i_{21}(\phi)^*\varphi\right)\left(x,y\right)-\left(\boldsymbol i_{22}(\phi)^*\varphi\right)\left(x,y\right) \\
&= \chi_{\Omega}\left(x,y\right)\dfrac{1}{\pi^2} \, p.v.\int_\Omega{\dfrac{\cos{\left[\phi(x',y')-\phi(x,y)\right]}\left[\varphi(x',y')-\varphi(x,y)\right]}{(x'-x)(y'-y)}\, dx'  \ dy'} \\
&-\chi_\Omega(x,y)\dfrac{1}{2\pi^4} \, p.v.\int_D\int_\Omega{\dfrac{\sin{\left[\phi(x,y)-\phi(x'',y'')\right]}\varphi(x',y')}{(x'-x)(y'-y)(x''-x')(y''-y')} \, dx'' \  dy''  \ dy'  \ dx'} \\
&-  \chi_\Omega(x,y)\dfrac{1}{2\pi^4}\, p.v. \int_D\int_\Omega{\dfrac{\sin{\left[\phi(x,y)-\phi(x'',y'')\right]}\varphi(x',y')}{(x'-x)(y'-y)(x''-x')(y''-y')} \, dx''  \ dy''  \ dy'  \ dx' } \\
&= \left(\tilde{\boldsymbol i_1}'(\phi)\ \varphi\right) \left(x,y\right)\\
&-\chi_\Omega(x,y)\dfrac{1}{\pi^4}\, p.v. \int_D\int_\Omega{\dfrac{\sin{\left[\phi(x,y)-\phi(x'',y'')\right]}\varphi(x',y')}{(x'-x)(y'-y)(x''-x')(y''-y')} \, dx'' \  dy''  \ dy' \  dx'}
\end{align*}
proving the claim of the proposition.
\end{proof}

As a consequence of the above proof the following remark is obtained.

\begin{remark}
    The operator $\tilde{\boldsymbol i_1}'(\phi):L^2\left(\R^2\right)\rightarrow L^2\left(\R^2\right)$ defined in eq.~\eqref{split2} is self-adjoint. Moreover, the linearization $i^{lin}$ defined in eq.~\eqref{liniQuad} is self-adjoint in $L^2(\R^2)$. 
\end{remark}

The self-adjointness of the above $L^2$-Fr\'echet derivative and linear intensity operators is, to our knowledge, a unique characteristic of the iQuad WFS and has not been found in other wavefront sensing techniques so far. This property makes the application of iterative reconstruction designs more straightforward.


\section{Double iQuad WFS}\label{double_iQuad}

Some parts of the wavefront have been found to be poorly seen by the iQuad WFS~\cite{FaHuShaRaLAM19_AO4ELTproc}. Affected are phases that contain only pure $x$- or $y$-frequencies, in other words, frequencies that lie on the edge of the Cartesian tessellation. For example, the iQuad sensor has difficulty coding the vertical astigmatism $Z_2^2$. The physical reason for this very low sensitivity is related to symmetry. 

Geometrical considerations suggest that a two-path optical system that simultaneously uses two iQuad sensors with two different orientations can circumvent the problem of poorly-seen modes.

The concept follows the double roof WFSs~\cite{Veri04}. By using a beam splitter, the field is divided into two fields which are Fourier-filtered by two different masks, here the iQuad mask and its conjugate the $-$iQuad mask, see Fig.~\ref{doubleiQuad}.

\begin{figure}[htbp]
\centerline{\includegraphics[width=6.5cm]{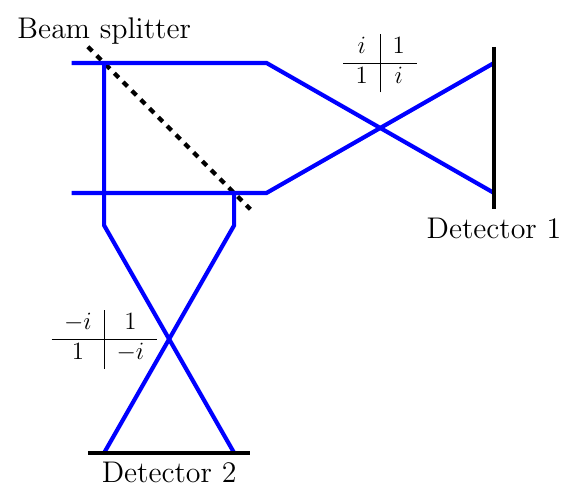}}
\caption{Optical setup of the double iQuad WFS.}
\label{doubleiQuad}
\end{figure}

\noindent The optical propagators of the two paths are
\begin{equation*}
A^+\stackrel{\eqref{Aplus}}{=}e^{\frac{i\pi}{4}}\left[\frac{\mathcal{I}+i H}{\sqrt{2}}\right]\, , ~~~~~~~~A^-=e^{-\frac{i\pi}{4}}\left[\frac{\mathcal{I}-i H}{\sqrt{2}}\right]\, ,
\end{equation*}
where $A^+$ corresponds to $\delta=\frac{\lambda}{4}$ and $A^-$ to $\delta=-\frac{\lambda}{4}$ in eq.~\eqref{fqpm_general}.

The incoming field amplitude is divided by a factor of $\sqrt{2}$ by the beam splitter. This results in two intensities
\begin{align*}
I^+(\phi)(x,y) &= \chi_D\left(x,y\right)\left|A^+\left(\frac{\chi_\Omega e^{-i\phi}}{\sqrt{2}}\right)(x,y)\right|^2\, , \\
I^-(\phi)(x,y) &= \chi_D\left(x,y\right)\left|A^-\left(\frac{\chi_\Omega e^{-i\phi}}{\sqrt{2}}\right)(x,y)\right|^2\, .
\end{align*}
Similarly to~\eqref{i1}, we obtain
\begin{align*}
    I^+(\phi)(x,y) 
    &= \chi_D\left(x,y\right)\frac{1}{2}\left[Im\left(\left(\chi_\Omega e^{-i\phi}\right)H\left(\chi_\Omega e^{i\phi}\right)\right)+\frac{\left|H\left(\chi_\Omega e^{i\phi}\right)\right|^2+\chi_\Omega^2}{2}\right](x,y)\, , \\
I^-(\phi)(x,y) 
    &= \chi_D\left(x,y\right)\frac{1}{2}\left[-Im\left(\left(\chi_\Omega e^{-i\phi}\right)H\left(\chi_\Omega e^{i\phi}\right)\right)+\frac{\left|H\left(\chi_\Omega e^{i\phi}\right)\right|^2+\chi_\Omega^2}{2}\right](x,y)\, .
\end{align*}

The double iQuad WFS output $dI$ is defined as the difference between these two intensities
\begin{align*}
dI(\phi)(x,y)&=I^+(\phi)(x,y)-I^-(\phi)(x,y) \\
&=\chi_D\left(x,y\right)\left[Im\left(\left(\chi_\Omega e^{-i\phi}\right)H\left(\chi_\Omega e^{i\phi}\right)\right)\right](x,y)\\
&\stackrel{\eqref{def_i11}}{=}\left(\boldsymbol i_1\phi\right)(x,y) \, .
\end{align*}
This results in
\begin{equation}\label{form_double_iquad}
dI\big(\phi\big)({x,y})=\chi_{\Omega}(x,y)\dfrac{1}{\pi^2} \, p.v.\int_\Omega{\dfrac{\sin{\left[\phi(x',y')-\phi(x,y)\right] }}{(x'-x)(y'-y)}\, dx'\ dy'}\, .
\end{equation}
The intensities of the iQuad WFS, the $-$iQuad WFS and the double iQuad WFS for an example incoming wavefront are visualized in Fig.~\ref{double_iquad}.

\begin{figure}[ht]
    \centering
    \includegraphics[width=0.24\linewidth]{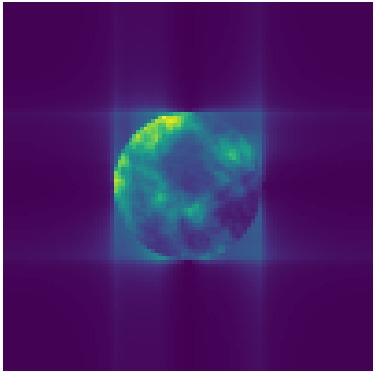}
    \includegraphics[width=0.24\linewidth]{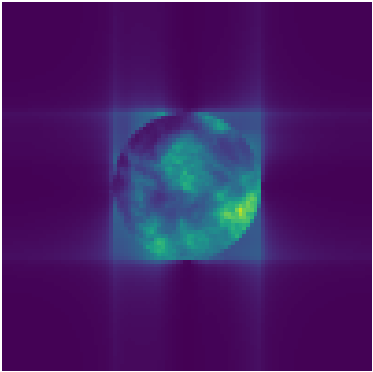}
    \includegraphics[width=0.24\linewidth]{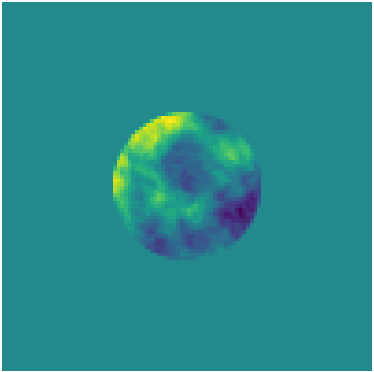}
    \caption{iQuad WFS intensity $I^+(\phi)$, $-$iQuad WFS intensity $I^-(\phi)$ and double iQuad WFS intensity $dI(\phi)$ (from left to right) for the example wavefront visualized in Fig.~\ref{ref_intensities}, left.}
    \label{double_iquad}
\end{figure}

Compared to the standard iQuad WFS, the underlying mathematical model of the two-path optical design is simplified as all intensity patterns related to $\boldsymbol i_2$ in Definition~\ref{def_i} vanish. Thus, the new optical configuration allows to cancel even intensity components of the iQuad WFS model. In particular, the first non-linear term is the cubic one and not, as usually, the quadratic one. This fact implies an improvement of the sensor linearity. 
Due to the identity of the intensity operators $dI = \boldsymbol i_1$, all derivations about $\boldsymbol i_1$ in Section~\ref{mathmodels} apply directly to the double iQuad WFS. The linearization of the double iQuad WFS operator $dI$ is given by $\boldsymbol i^{lin}$ defined in eq.~\eqref{liniQuad} and corresponds to the linear intensity of the standard iQuad WFS. Moreover, the linear double iQuad operator $\boldsymbol i^{lin}$ is self-adjoint in $L_2\left(\R^2\right)$ as is the corresponding Fr\'echet derivative $\widetilde{dI}'(\phi):=\tilde{\boldsymbol i_1}'(\phi)$ from eq.~\eqref{split2}. This property is particularly advantageous for certain iterative methods, such as the linear Landweber iteration~\cite{HuRaSha19_2}.
Eq.~\eqref{form_double_iquad} reveals that all light on the detector is concentrated within the pupil geometry $\Omega$ for the double iQuad WFS. Pixels outside of $\Omega$ have no relevant intensity.


\section{Similarities of the (double) iQuad WFS to the PWFS}\label{sim_pwfs}

The PWFS is the most well-known Fourier-type WFS~\cite{Raga96}. It is included in the design of many current and upcoming AO systems, e.g., in first-light instruments of the ELT. The mathematical model of the PWFS corresponds to the general Fourier-type WFS model presented in eq.~\eqref{eq_fouriertype} but with a 4-sided pyramidal prism as OTF in the focal plane.

The mathematical model of the double iQuad WFS and the slope definition of the PWFS~\cite{Veri04} are closely related. The latter consists of two specific combinations of the four pupil images and allows to numerically compress the PWFS signal into two measurements. Assuming a reflective PWFS working in its linear regime, the slope maps $\left[s_x,s_y\right]$ are approximated by~\cite{HuRaSha19_1}
\begin{align}
s_x\big(\phi\big)(x,y) &\approx -\chi_\Omega(x,y) \frac{1}{2\pi}\int_{\Omega_y} \frac{\phi(x',y)-\phi(x,y)}{x'-x}\, dx'\, , \label{Sx} \\ 
s_y\big(\phi\big)(x,y) &\approx \chi_\Omega(x,y) \frac{1}{2\pi}\int_{\Omega_x} \frac{\phi(x,y')-\phi(x,y)}{y'-y}\, dy'\, , \label{Sy}
\end{align}
where $\Omega_x = \{(\tilde{x},\tilde{y}) \in \Omega, \tilde{x}=x\}$ represents a single line of $\Omega$ for fixed $x$ and $\Omega_y$ is defined respectively.  

The above quantities can be viewed as integrated difference quotients along the $x$- and $y$-directions, which is why $\left[s_x,s_y\right]$ are often interpreted as phase derivatives (“slopes”). However, they are not true phase derivatives. The 1d Hilbert transform along $x$ (and $y$ respectively) is the dominant operation in the measurements $s_x$ and $s_y$.

We now compare eq. \eqref{Sx}-\eqref{Sy} to the linear intensity of the (double) iQuad WFS given in eq.~\eqref{liniQuad} as
\begin{equation*}
\left(\boldsymbol i^{lin}\phi\right)(x,y)= \chi_{\Omega}(x,y)\dfrac{1}{\pi^2} \, p.v. \int_\Omega{\dfrac{\phi(x',y')-\phi(x,y)}{(x'-x)(y'-y)}\, dx' \ dy'}\, .
\end{equation*}

The iQuad WFS signal appears to be a condensed version of both PWFS slope maps. Specifically, the iQuad WFS optically encodes the phase's $x$- and $y$-Hilbert transforms  within a single pupil image. Although a focal-plane pyramidal prism differs from the iQuad’s OTF, their connection stems from a shared Cartesian tessellation of the focal plane. Regardless of how the phase information is processed — whether through optical interference via the differential piston in the iQuad WFS or by splitting spatial frequencies with prisms and numerically processing them in the PWFS - the underlying mathematical operators are fundamentally determined by the way in which the focal plane is subdivided.

The similarity between the iQuad WFS and the slope maps of the PWFS offers several advantages. The wealth of existing model-based reconstruction algorithms developed for the PWFS can be readily adapted to create efficient reconstructors for the iQuad WFS~\cite{ShaHuRa20}. 

The iQuad WFS can also be viewed as a theoretical extension of the PWFS that is both optically and mathematically simpler. Its mask is simpler and avoids design assumptions such as an optimal pyramid angle. In the (double) iQuad configuration, the Fréchet derivative and the linear-intensity operator are self-adjoint, unlike in the PWFS, further clarifying its analysis and implementation.

Moreover, the iQuad WFS is more efficient in terms of detector size than the PWFS. This can be analyzed from two perspectives.

Firstly, we use geometrical optics, assuming that phase information is confined to pupil images only. The iQuad encodes both $x$- and $y$-information in a single pupil image, whereas the PWFS produces four pupil images. Thus, the iQuad needs about four times fewer pixels for equivalent information, reducing readout and processing.

Secondly, we examine their linear intensity. For optimal AO correction using conventional interaction matrix-based methods~\cite{Ellerbroek02}, it is essential to maximize the use of sensor pixels within the linear intensity region. Eq.~\eqref{liniQuad} shows that the iQuad WFS's linear intensity region matches exactly the geometrical pupil image, whereas the PWFS's linear information extends beyond the four pupil images, favoring use of the full detector. As a result, the PWFS typically requires approximately four times the number of pixels, confirming the superior detector efficiency of the iQuad WFS.


\section{The modulated iQuad WFS}\label{mod_iquad}

Similarly to the PWFS, it is possible to extend the iQuad WFS concept with modulation~\cite{Raga96}. Following the notation in~\cite{Fauv16}, we use a weighting function $w$ to characterize the shape of the modulation. It is a 2d function that codes the time spent for each modulation cycle. The function $w$ is thus real, positive, normalized with respect to the $L^1$-norm and symmetric around its center. The latter assumption is physically supported as the iQuad mask is also symmetric around its center and the incoming wavefront has an isotropic structure.

Mathematically, it is observed that modulation implies the change~\cite{Fauv16}
\begin{equation*}
\frac{1}{\pi^2 x y} \rightarrow \frac{\hat{w}(x,y)}{\pi^2 x y}\, ,
\end{equation*}
where 
$\hat{w}$ represents the 2d Fourier transform of the modulation weighting function $w$.
Thus, for the modulated iQuad WFS the corresponding meta-intensity $m^{mod}$ reads as

$$m^{mod}(x,y)\boldsymbol = \boldsymbol i^{mod} = \boldsymbol i_1^{mod}+\boldsymbol i_2^{mod}$$
with
\begin{align*}
\left(\boldsymbol i_1^{mod}\phi\right)(x,y)&:=\chi_{\Omega}(x,y)\dfrac{1}{\pi^2}\, p.v. \int_\Omega{\dfrac{\sin{\left[\phi(x',y')-\phi(x,y)\right] }}{(x'-x)(y'-y)}\, \hat{w}\left(x'-x,y'-y\right)\, dx'\ dy'}\, , \\
\left(\boldsymbol i_2^{mod}\phi\right)(x,y)&:= \chi_D\left(x,y\right)\dfrac{1}{2\pi^4}\, p.v.  \int_\Omega\int_\Omega\dfrac{\cos{\left[\phi(x',y')-\phi(x'',y'')\right]-1 }}{(x'-x)(y'-y)(x''-x)(y''-y)}\\
&\cdot \hat{w}\left(x'-x'',y'-y''\right) \, dx''\ dy''\  dy'   \ dx' \, .
\end{align*}

The linearization around $0$ of the operator $i^{mod}$ is given by
\begin{equation*}
\left(\boldsymbol i^{mod,lin}\phi\right)(x,y):= \chi_{\Omega}(x,y)\dfrac{1}{\pi^2} \, p.v. \int_\Omega{\dfrac{\phi(x',y')-\phi(x,y)}{(x'-x)(y'-y)} \hat{w}\left(x'-x,y'-y\right) \, dx'  dy'} \, .
\end{equation*}

As stated in article~\cite{Veri04}, modulation can increase the linearity of the WFS at the cost of sensitivity.

\section{Conclusion \& outlook}

This work presents the mathematical foundation for a new Fourier-type WFS called the iQuad WFS, originally derived from the FQPM coronagraph. Additionally, a two-path variation, the double iQuad WFS, has been proposed. This concept utilized two iQuad masks with different orientations. The nonlinear integral operators representing both the iQuad and double iQuad WFS were analyzed and their Fr\'{e}chet derivatives were calculated. Linearizations of these operators along with their corresponding adjoint operators were presented. Furthermore, it was demonstrated how the mathematical model of the double iQuad WFS establishes a connection to the well-known PWFS due to the similarity between the iQuad WFS output and the PWFS slope maps. Finally, the concept of modulation was introduced in the context of iQuad wavefront sensing.

This theoretical analysis serves as a foundational step toward applying nonlinear and linear model-based reconstruction algorithms for iQuad wavefront sensing in AO, which will be detailed in a follow-up paper. Future work will focus on laboratory demonstrations of the iQuad WFS, including extensive comparisons with existing WFSs such as the PWFS.

It was mentioned that there are some phase components that might be poorly seen by the iQuad WFS. This problem is addressed by the double iQuad WFS design. However, an important question is how likely it is, in real-world scenarios, to encounter perfectly symmetric spatial frequency components of a wavefront being poorly detected by the iQuad WFS. While such modes may pose challenges in simulations, they might rarely, if ever, occur in practical applications, such as on-sky observations or medical imaging, simply because these perfectly symmetric wavefront components might be typically absent. Therefore, further laboratory investigation is needed to determine whether a single-path iQuad WFS performs adequately in practice, or whether the double iQuad WFS is a more reliable choice.

\paragraph{Funding.} 
This research was funded in part by the Austrian Science Fund (FWF) SFB 10.55776/F68 "Tomography Across the Scales", Project F6805-N36 (Tomography in Astronomy). For open access purposes, the authors have applied a CC BY public copyright license to any
author-accepted manuscript version arising from this submission. Alfred Miksch is partially supported by the State of Upper Austria. This work also benefited from the support of the French National Research Agency (ANR) with the Programme Investissement Avenir F-CELT (ANR-21-ESRE-0008), the ANR-DGA-AID ASTRID program (ANR-25-ASTR-0015), the Action Sp\'ecifique Haute R\'esolution Angulaire (ASHRA) of CNRS/INSU co-funded by CNES, the French Government under the France 2030 investment plan (cassiop\'ee project) and the Initiative d’Excellence d’Aix-Marseille Universit\'e A*MIDEX, program number AMX-22-RE-AB-151.


\bibliographystyle{plain}
\bibliography{iquad}

\begin{appendix}
    \section{Appendix}

\begin{theorem}
\label{theorem: Fubini L^1(Omega^3)}
    Let $\Omega, \Omega',\Omega''$ be subsets of $\R^2$ and $f\in L^1(\Omega\times\Omega'\times\Omega'')$. Assume that
    \begin{align*}
        &p.v. \int_{\Omega'}\int_{\Omega''} \frac{f(x,y,x',y',x'',y'')}{(x'-x)(y'-y)(x''-x)(y''-y)} \ dx'' \ dy''\ dx'\ dy'\, \\
        \text{and}\quad & p.v. \int_{\Omega'}\int_\Omega\frac{f(x,y,x',y',x'',y'')}{(x'-x)(y'-y)(x''-x)(y''-y)} \ dx \ dy\ dx'\ dy'
    \end{align*}
    exist for almost every $(x,y)\in\Omega$ and $(x'',y'')\in\Omega''$, respectively. Then, it holds that
    \begin{align}
    \begin{split}
        &\int_\Omega p.v. \int_{\Omega'}\int_{\Omega''} \frac{f(x,y,x',y',x'',y'')}{(x'-x)(y'-y)(x''-x)(y''-y)} \ dx'' \ dy''\ dx'\ dy'\ dx\ dy\\
        &= \int_{\Omega''} p.v. \int_{\Omega'}\int_\Omega\frac{f(x,y,x',y',x'',y'')}{(x'-x)(y'-y)(x''-x)(y''-y)} \ dx \ dy\ dx'\ dy' \ dx''\ dy''\, ,
    \end{split}
        \label{eq: "Fubini" for p.v. integrals}
    \end{align}
    if at least one of the integrals exists.
\end{theorem}

\begin{proof}
    For an $\epsilon >0$, the integrand
    \begin{equation*}
       g(x,y,x',y',x'',y''):=\frac{f(x,y,x',y',x'',y'')}{(x'-x)(y'-y)(x''-x)(y''-y)}
    \end{equation*}
    is in $L^1(\Omega_\epsilon^3)$, where
    \begin{equation*}
        \Omega^3_\epsilon :=\{(x,y)\in\Omega,(x',y')\in\Omega',(x'',y'') \in \Omega'' \ : \ |x'-x|>\epsilon , \ |y'-y|>\epsilon,\ |x''-x|>\epsilon,\ |y''-y|>\epsilon\}\, ,
    \end{equation*}
    since the denominator of $g$ is bounded in $\Omega^3_\epsilon$, i.e., $\|g\|_{L^1(\Omega_\epsilon^3)}\leq \epsilon^{-4} \|f\|_{L^1(\Omega\times\Omega'\times\Omega'')}$.
    
    Thus, due to Fubini, the following equation holds
    \begin{align}
    \begin{split}
        &\int_\Omega\underbrace{\int_{\Omega'} \int_{\Omega''} \chi_{\Omega^3_\epsilon} \ g(x,y,x',y',x'',y'')\ dx''\ dy''\ dx' \ dy'}_{=: I_\epsilon(x,y)} \ dx\ dy \\
        &= \int_{\Omega''} \underbrace{\int_{\Omega'}\int_\Omega \chi_{\Omega^3_\epsilon} \ g(x,y,x',y',x'',y'')\ dx \ dy\ dx'\ dy'}_{=:J_\epsilon(x'',y'')} \ dx''\ dy''\, .
    \label{eq: int I_eps = int J_eps Omega3}\end{split}\end{align}
    
    Without loss of generality, it is assumed that the first integral in \eqref{eq: "Fubini" for p.v. integrals} exists.
    Using $g^+ := \chi_{\{g\geq 0\}}g$ and $g^- := -\chi_{\{g<0\}}g$, we can express $I_\epsilon$ as
    \begin{align*}
        I_\epsilon(x,y) &= I_\epsilon^+(x,y) - I_\epsilon^-(x,y)\, ,\\
        I_\epsilon^+(x,y)&:=  \int_{\Omega'}\int_{\Omega''} \chi_{\Omega^3_\epsilon}\ g^+(x,y,x',y',x'',y'') \ dx''\ dy''\ dx'\ dy'\, ,\\
        I_\epsilon^-(x,y)&:=  \int_{\Omega'} \int_{\Omega''} \chi_{\Omega^3_\epsilon}\ g^-(x,y,x',y',x'',y'') \ dx''\ dy''\ dx'\ dy'\, .
    \end{align*}
    Since $g^+$ and $g^-$ are non-negative, $I_\epsilon^+$ and $I_\epsilon^-$ are also non-negative, and both of them are monotonically increasing a.e. as $\epsilon \rightarrow0$, i.e,
    \begin{equation*}
        \epsilon_1 \leq \epsilon_2 \Rightarrow \left[I_{\epsilon_1}^+ \geq I_{\epsilon_2}^+ \ \land \  I_{\epsilon_1}^- \geq I_{\epsilon_2}^-\right] \text{ a.e.}\, .
    \end{equation*}
    Thus, we can apply the monotone convergence theorem (MCT) and get
    \begin{align*}
    &\int_\Omega p.v. \int_{\Omega'}\int_{\Omega''} \frac{f(x,y,x',y',x'',y'')}{(x'-x)(y'-y)(x''-x)(y''-y)} \ dx'' \ dy''\ dx'\ dy'\ dx\ dy\\
    & = \int_\Omega \lim_{\epsilon\rightarrow 0} I_\epsilon(x,y) \ dx\ dy = \int_\Omega \lim_{\epsilon\rightarrow 0} I^+_\epsilon(x,y) - \lim_{\epsilon\rightarrow 0} I^-_\epsilon(x,y) \ dx\ dy \\
    & \stackrel{MCT}{=} \lim_{\epsilon\rightarrow 0} \int_\Omega I^+_\epsilon(x,y) \ dx\  dy- \lim_{\epsilon\rightarrow0}\int_\Omega I^-_\epsilon(x,y) \ dx\  dy =\lim_{\epsilon\rightarrow 0} \int_\Omega I_\epsilon(x,y) \ dx\ dy \, .
    \end{align*}
    Because of \eqref{eq: int I_eps = int J_eps Omega3}, we know that
    \begin{equation*}
        \lim_{\epsilon\rightarrow 0} \int_\Omega I_\epsilon(x,y) \ dx\ dy = \lim_{\epsilon\rightarrow 0} \int_{\Omega''} J_\epsilon(x'',y'')\ dx''\ dy''\, ,
    \end{equation*}
    in particular, the latter limit exists.
    
    Using analogous arguments as before, one can show that
    \begin{equation*}
        \lim_{\epsilon\rightarrow 0} \int_{\Omega''} J_\epsilon(x'',y'')\ dx''\ dy'' = \int_{\Omega''} p.v. \int_{\Omega'}\int_\Omega\frac{f(x,y,x',y',x'',y'')}{(x'-x)(y'-y)(x''-x)(y''-y)} \ dx \ dy\ dx'\ dy' \ dx''\ dy''\, ,
    \end{equation*}
    
    which yields \eqref{eq: "Fubini" for p.v. integrals}.
\end{proof}

Now, we look into the exchange of limits in the p.v. meaning and integrals in the proof of Proposition~\ref{roof_frechet_adjoint}. It holds that the function
\begin{align*}
    &h:\quad (x',y',x'',y'') \mapsto \sin{[\phi(x',y')-\phi(x'',y'')]}
\end{align*}
is in $L^\infty(\R^4)$ with $\|h\|_{L^\infty(\R^4)} = 1$. 
For $\varphi,\psi\in L^2(\R^2)$ and $D, \Omega$ bounded, we get with the Cauchy-Schwarz (C.S.) inequality
\begin{align*}
    &\int_D\int_\Omega \int_\Omega |h(x',y',x'',y'')\psi(x',y')\varphi(x,y)| \, dx''\ dy'' \ dy'\ dx'\ dy\ dx\\
    & \leq |\Omega|\|h\|_{L^\infty(\R^4)} \left( \int_\Omega \left|\psi(x',y')\right|\, dy'\ dx' \right)\left(\int_D \left|\varphi(x,y)\right| \, dy\ dx\right)\\
    & \overset{C.S.}{\leq}|\Omega|^{3/2} |D|^{1/2}\|\psi\|_{L^2(\Omega)} \|\varphi\|_{L^2(D)} < \infty\, .
\end{align*}
As a result, the function
\begin{equation*}
    (x,y,x',y',x'',y'') \mapsto \sin\left[\phi(x',y')-\phi(x'',y'')\right] \psi(x'',y'')\varphi(x,y)\, 
\end{equation*}
is in $L^1(D\times\Omega\times\Omega)$ and Theorem~\ref{theorem: Fubini L^1(Omega^3)} can be applied. Analogous results can be shown for double integrals also included in the proof of Proposition~\ref{roof_frechet_adjoint}.

\end{appendix}

\end{document}